\documentclass[]{article}
\pdfoutput=1

\usepackage[none]{hyphenat}
\usepackage[margin=1.5in]{geometry}

\usepackage{natbib}
\usepackage{amsmath}
\usepackage{amsthm}
\usepackage{amssymb}
\usepackage{amsfonts}
\usepackage{mathrsfs}

\usepackage[colorlinks=true,linkcolor=red,citecolor=red,filecolor=red,urlcolor=red]{hyperref}

\newtheorem{theorem}{Theorem}
\newtheorem{lemma}[theorem]{Lemma}
\newtheorem{corollary}[theorem]{Corollary}

\newtheorem{definition}[theorem]{Definition}

\usepackage{epsfig}
\usepackage{verbatim}
\usepackage{minibox}

\def\F {\ensuremath{\mathbb{F}}}

\def\P {\ensuremath{\mathbb{P}}}
\def\S {\ensuremath{\mathbb{S}}}
\def\dirsum{\oplus}

\newcommand{\Pbr}[2] {\ensuremath{P_{q,#1,#2}}}
\DeclareMathOperator{\diag}{diag}
\DeclareMathOperator{\snf}{snf}

\DeclareMathOperator{\minpoly}{min-poly}

\DeclareMathOperator{\rank}{rank}

\begin{document}

\title{Probabilistic Analysis of Block Wiedemann for Leading Invariant Factors}

\author{Gavin Harrison \\
Drexel University \and
Jeremy Johnson\\
Drexel University \and
B. David Saunders\\
University of Delaware}

\maketitle

\begin{abstract}
We determine the probability, structure dependent, that the block Wiedemann algorithm correctly computes leading invariant factors.  This leads to a tight lower bound for the probability, structure independent.
We show, using block size slightly larger than $r$, that the leading $r$ invariant factors
are computed correctly with high probability over any field.  Moreover, an algorithm is provided to
compute the probability bound for a given matrix size and thus to select the block size
needed to obtain the desired probability.  The worst case probability bound is improved,
post hoc, by incorporating the partial information about the invariant factors.
\end{abstract}


\section{Introduction}

For prime power $q$, let $\F$ denote a finite field of cardinality $q$.  
Let $A$ be a $n\times n$ matrix over $\F$. 
For chosen block size $b$, let $U,V$ be uniformly random in $\F^{n \times b}$.
Call the sequence $S = \{U^T A^i V\}_{i}, i \in \{0,1,\ldots\}$ the {\em (U,V)-projection} of $A$.  
The Wiedemann ($b=1$) and Coppersmith's block Wiedemann ($b>1$) algorithms compute the minimal generating polynomial, $G \in \F[x]^{b \times b}$, of $S$.
This means that 
$\sum_{i=0}^d S_{k+i} G_i = 0$, 
 for all $k\geq 0$, where $d = \deg(\det(G)) > 0$, 
and $d$ is minimal
\citep{Wiedemann86,Coppersmith94}.  
Then
the $i$-th largest invariant factor of $G$ divides the $i$-th largest invariant factor of $xI-A$, and is equal with high probability \citep{KaVi:2001, KaVi:2004} for large enough field.  
Observations of the behavior for small fields were noted by Coppersmith and the analysis has been extended by \cite{Villard97Further, Villard97} and \cite{Brent:2003} to small fields subject to certain constraints.
We call a projection and it's minimal generator, $G$, {\em r-faithful} to $A$ if the $r$ largest invariant factors of $G$ are the $r$ largest invariant factors of $xI-A$.  

Wiedemann and Coppersmith developed their algorithms for the purpose of solving linear systems and weren't explicitly concerned with determining invariant factors.
Prior analysis of the block Wiedemann algorithm was motivated by this problem, and is therefore one-sided (asymmetric in the treatment of projection from left and right).  Given $X \in \F^{b \times n}$ and $Y \in \F^{n \times m}$ chosen uniformly at random, where $b \geq m$, \cite{Villard97Further, Villard97} gives a bound on the probability that the minimal generating polynomial of $\{XA^iY\}_i$ is $m$-faithful to the minimal generating polynomial of $\{A^iY\}_i$.  An exact formula and tighter bound for this probability are given in \cite{Brent:2003}.  These analyses are dependent on $b \geq m$ and the minimal generator of $\{A^iY\}_i$ having at most $m$ nontrivial invariant factors, thus eliminating the ``pathological'' case discussed in \cite{Coppersmith94}.  They do not speak directly to two sided analysis for situations in which $xI-A$ has more than $b$ nontrivial invariant factors.  Moreover, it is important for our purposes to have a probability bound that applies without regard to the matrix structure and that quantifies the increased confidence that can be achieved for computing $r$ invariants by selecting block size $b$ somewhat larger than $r$.

In this paper we develop an exact formula (Theorem \ref{thm:probrec}), $\Pbr{b}{r} (A)$, for the probability that a random projection is $r$-faithful, for given eigenstructure of $A$.  
We then construct the worst case  and derive a sharp lower bound (Theorem \ref{theorem:worst}), $P_{q,b,r}(n) = \min_{A \in \F^{n \times n}} \Pbr{b}{r}(A)$, on the probability that a random projection is $r$-faithful for arbitrary $n\times n$ matrix $A$.
Since the worst case occurs when there are exactly $r$ invariant factors, the bounds from \cite{Brent:2003} can be applied to estimate $P_{q,b,r}(n)$ (Theorem \ref{theorem:brent}).
Knowing $P_{q,b,r}(n)$ allows $b$ to be computed such that $P_{q,b,r}(n) \geq p$ for any desired probability $p$.  Using this we show that with a block size slightly larger than $r$ the projection
is $r$-faithful with high probability.  This makes precise previous observations and estimates
regarding block size.  The results in this paper are
an extension of our previous work in which we presented formulas for $\Pbr{b}{1}(A)$ and $P_{q,b,1}(n)$ \citep{HJS:2016}.  

The worst case bound can be improved by incorporating information about the invariant factors
of the minimal generating matrix $G$.  In the extreme case, where the sum of the degrees of 
the invariant factors of $G$ equals the matrix dimension, the invariant factors of $G$ are equal
to those of $xI-A$.  In less extreme cases the partial information obtained from $G$
can be used, post hoc, to improve the probability bounds for $G$ to be $r$-faithful (Theorem \ref{thm:post-hoc}).

The main results of this paper have been presented without the proofs as a poster at ISSAC 2016, with abstract \citep{HJS:CCA:2017}.
They are presented here with full development and proofs along 
with examples and new results on the post hoc analysis.

\section{Probability Analysis}

In this section we derive and prove an exact formula (Theorem \ref{thm:probrec}), $\Pbr{b}{r} (A)$, for the probability that a random projection is $r$-faithful, for given eigenstructure of $A$.
Similarly to the proofs in \cite{Villard97Further, Villard97} and \cite{Brent:2003}, our analysis reduces the probability calculation, first to primary components and then to a direct sum of companion matrices of irreducible polynomials.  

After introducing notation and some technical results, we show that the probability calculation can be split into independent consideration of the distinct primary components (Theorem \ref{thm:probprod}).
Then the probability for a primary component is reduced to that of a direct sum of companion matrices (Theorem \ref{thm:prob_lif}).  Finally, we show that the sequences generated by the individual companion matrices can mapped to vector outer products (Lemma \ref{lemma_outer}), which reduces the problem to a rank calculation (Theorems \ref{thm:prob_companion} and \ref{thm:probrec}).

The following notation will be used throughout the paper.  Starting with finite field $\F$, we will be working with the ring of polynomials $\P = \F[x]$ and it's modular images $\P_f = \F[x]/\langle f\rangle$, ($f \in \P$).
Let $A \in \F^{n \times n}$, $U, V \in \F^{n \times b}$. We are concerned with matrix sequences of the form $S = \{U^TA^iV\}_{i=0}^\infty$.  Define an action of $\P$ on such sequences by, for polynomial $f$ of degree $d$, $S \rightarrow fS$, where $(fS)_k = \sum_{i=0}^d f_iS_{k+i}$.  
When $fS = 0$ we say that $f$ generates $S$. 

Let $\S_f^{b \times b}$ denote the set of $b \times b$ matrix sequences generated by the scalar polynomial $f \in \P$. $\S_f^{b\times b}$ is a module over $\P$ with respect to the action given above.  
For short we use $\S_f = \S_f^{1\times 1}$ to denote the scalar sequences.
Let $\P^{b \times b}$ denote $b \times b$ matrix polynomials
and similarly for $\P_f^{b \times b}$.
We define a mapping, $\phi_f : \S_f^{b \times b} \rightarrow \P_f^{b \times b}$ that is both a vector space isomorphism over $\F$ and isomorphism of right $\P$ modules.  It also satisfies the property $\phi_f(SG) = \phi_f(S)G_f$, for $G_f$ being $G$ modulo $f$ (this is Corollary \ref{corr:map}).  
Note that $\P_f^{b\times b}$ is a $\P$-module and 
for $G,H \in \P^{b\times b}$, we have
$\phi_f(S)G = \phi_f(S)H$ whenever 
$H \equiv G (\mbox{mod }f)$.  This mapping is an extension of the mapping used by 
\cite{Wiedemann86} in his probabalistic analysis to $b \times b$ blocks with the details 
made explicit.

Now consider the action from the right of $\P^{b\times b}$ on matrix sequences.
We say $G \in \P^{b \times b}$ generates $S$, if $\det(G) \neq 0$ and $(SG)_k = \sum_{i=0}^d S_{k+i} G_i = 0$ for all $k\geq 0$, where $d = \deg(G)$.  
$G$ is minimal if its columns 
form a basis for the annihilator 
of $S$, or equivalently $\deg(\det(G))$ is minimal \citep{Yuhasz13}.
It follows that two minimal generators for $S$ are unimodularly equivalent and thus have the same Smith normal form.
By Theorem 2.12 in \citet{KaVi:2004}, if $G$ is minimal then the $i$-th invariant factor of $G$ divides the $i$-th invariant factor of $xI-A$.  This section analyzes the probability that for random $U$ and $V$ that $G$ is $r$-faithful to $A$.

\begin{definition}
Let $A \in \F^{n \times n}$, let $U,V \in \F^{n \times b}$ be uniformly random, 
let $S = \{U^TA^iV\}_{i=0}^\infty$, and let $G \in \P^{b \times b}$ minimally generate projection $S$.  
Define $\Pbr{b}{r}(A)$ to be the probability that $G$ is $r$-faithful to $A$.
\end{definition}

Let $A = XJY$ where $X,Y \in \F^{n \times n}$ are nonsingular and $J \in \F^{n \times n}$ is a generalized Jordan normal form.  Because $U,V$ are chosen uniformly at random, $UX$ and $YV$ are also uniformly random and $\Pbr{b}{r}(A) = \Pbr{b}{r}(J)$.
Therefore, we can restrict our analysis to matrices in Jordan form: $A = \dirsum_{i,j} J_{f_i^{e_{i,j}}}$ where $f_i \in \P$ are distinct monic irreducible polynomials and $J_{f^e}$ denotes the generalized Jordan block associated with $f^e$.
Let $C_f$ denote the companion matrix of $f$, 
 and note that $C_f = J_{f^1}$.
$$ C_f = \begin{bmatrix} 
0 & 0 & \ldots & -f_0\\
1 & 0 & \ldots & -f_1\\
0 & \ddots & \ddots & \vdots\\
0 & 0 & 1 & -f_{d-1}\\
\end{bmatrix},~~~ 
J_{f^e} = 
\begin{bmatrix} 
C_f & 0 & \ldots & 0\\
I & C_f & \ldots & 0\\
0 & \ddots & \ddots & 0\\
0 & 0 & I & C_f\\
\end{bmatrix},$$ 

\begin{definition}
Let $f \in \P$ be a scalar polynomial of degree $d$.  Define $\rho : \P_f \rightarrow \F^{d \times d}, \rho(a)b = ab \mod f,$ 
to be the regular representation of the polynomial algebra
$\P_f$.
Here we are equating polynomials modulo $f$ with their column vectors of coefficients and, explicitly, $\rho(a)$ is the Krylov matrix $K_f(a)$ generated by the companion matrix $C_f$ and $a$. 
$$\rho(a) = K_f(a) = \sum_{i=0}^{d-1} a_iC_f^i = \begin{bmatrix}a & C_fa & \ldots & C_f^{d-1}a\end{bmatrix}.$$ 
Define 
$\omega_f : \S_f \rightarrow \P_f$ by $\omega_f(S) = \sum_{i=0}^{d-1}S_ix^i$, 
and then   
define
$\phi_f : \S_f \rightarrow \P_f$ by $\phi_f(S) = P\omega_f(S),$ where $P$ is a nonsingular matrix satisfying $P\rho^T(a) = \rho(a)P$ for all $a \in \P_f$.  The existence of such $P$ is shown 
in \citep{taussky1959}.
Extend $\omega_f$ and $\phi_f$ componentwise to $\S_f^{b\times b} \rightarrow \P_f^{b\times b}$.
\end{definition}


\begin{lemma}
\label{lem:mapping}
Let $f \in \P$ be a polynomial of degree $d$, $S \in \S_f$, and $g \in \P.$  Then $\phi_f(Sg) = \rho(g)\phi_f(S) = g\phi_f(S).$
\end{lemma}
\begin{proof}
Because $S$ is generated by $f$, and $\deg(f) = d$, the sequences $S$ and $Sg$ are fully defined by their first $d$ elements.  We can write $\omega_f(Sg)$ as a Hankel matrix times vector product and observe that
\begin{eqnarray*}
\omega_f(Sg) &=& \left[\begin{matrix}
S_0 & S_1 & \cdots & S_{d-1} \\
S_1 & S_2 &        & S_d \\
\vdots &  & \ddots & \vdots \\
S_{d-1} & S_{d} & \cdots & S_{2d-2}
\end{matrix}\right]\left[\begin{matrix}
g_0 \\ g_1 \\ \vdots \\ g_{d-1}
\end{matrix}\right]
= \left[\begin{matrix}
\omega_f^T(S) \\
\omega_f^T(S)C_f \\
\vdots \\
\omega_f^T(S)C_f^{d-1}
\end{matrix}\right]\left[\begin{matrix}
g_0 \\ g_1 \\ \vdots \\ g_{d-1}
\end{matrix}\right] \\
&=& \left(\omega_f^T(S)\left[\begin{matrix}g & C_fg & \cdots & C_f^{d-1}g\end{matrix}\right]\right)^T \\
& = &
(\omega_f^T(S)\rho(g))^T = \rho^T(g)\omega_f(S), 
\end{eqnarray*}
Therefore,
$\phi_f(Sg) = P\omega_f(Sg)$ (by definition), which by the previous observation equals 
$P\rho^T(g)\omega_f(S)$.
Using the definition of $P$ and $\rho$, we have 
$P\rho^T(g)\omega_f(S) = 
\rho(g)P\omega_f(S) = \rho(g)\phi_f(S) = g\phi_f(S)$.

\end{proof}

\begin{corollary}
\label{corr:map}
Let
$S \in \S_f^{b \times b}, G \in \P^{b \times b}.$  Then, $\phi_f(SG) = \phi_f(S)G$.
\end{corollary}
\begin{proof}
\begin{eqnarray*}
\phi_f(SG)_{ij} &=& \phi_f((SG)_{ij})
 = \phi_f\left(\sum_{k=0}^d S_{ik}G_{kj}\right)
 = \sum_{k=0}^d \phi_f(S_{ik}G_{kj}) \\
&=& \sum_{k=0}^d \phi_f(S_{ik})G_{kj}
 = (\phi_f(S)G)_{ij},
\end{eqnarray*}
where $d = \deg(G).$
\end{proof}

In view of Corollary \ref{corr:map},
$G$ generates $S$ if and only if $\phi_f(S)G = 0$ and $\det(G) \neq 0$.  
Motivated by this we will also speak of generating matrices over $\P_f$: $G$ generates $A \in \P_f^{b \times b}$ if $AG = 0$ and $\det(G) \neq 0$, and $G$ minimally generates $A$ if $\deg(\det(G))$ is minimal.  Lemma \ref{lem:min_gen_bound} and Theorem \ref{thm_pc_if} relate the Smith normal 
form (snf) of a matrix over $\P_f$ and the Smith normal form of its minimal generating matrix in the case that $f$ is an irreducible power.

\begin{lemma}
\label{lem:min_gen_bound}
Let $f \in \P$ be an irreducible polynomial of degree $d$ and let $e$ be a positive integer.  Let $A \in \P_{f^e}^{b \times b}$, and $r_i$ be the number of non-zero invariant factors of $Af^i$.  If $G \in \P^{b \times b}$ generates $A$, then $\deg(\det(G)) \geq \sum_{i=0}^{e-1} r_i d$.
\end{lemma}
\begin{proof}
Let $G = X \snf(G) Y$, where $X,Y \in \P^{b \times b}$ are unimodular. 
Let $g_j$ denote the number of invariant factors of $G$ divisible by $f^j$ and let 
$G_j$ denote\\ $diag(1,\ldots,1,f,\ldots,f)$, in which $f$ is repeated $g_j$ times.  
Since $Af^e=0$, we have that $\snf(G) = \prod_{j=0}^{e-1} G_j$, with $g_j \geq g_{j+1}$ being the count of invariant factors equal to $f^j$. 
Moreover, since $G$ generates $A$, $AX \snf(G) Y = AX \prod_{j=0}^{e-1} G_jY =0$, and 
since $Y$ is unimodular, $ AX \prod_{j=0}^{e-1} G_j =0$.
Using this as the base case, it follows by induction that 
$Af^i X \prod_{j=i}^{e-1} G_j = 0$, for $0 \le i < e$.
Since $Af^i X \prod_{j=i}^{e-1} G_j = 0$ and $G_i$ has $b-g_i$ ones 
followed by $g_i$ copies of $f$ along the diagonal and $g_i \geq g_j$ for $j=i+1,\ldots,e-1$, it follows that
$Af^i X$ is a matrix whose first $b-g_j$ columns are zero.
\[
Af^i X =
\begin{pmatrix} 
  \overbrace{\begin{matrix}0\end{matrix}}^{b-g_i} & \vline &
  \overbrace{\begin{matrix}*\end{matrix}}^{g_i}
\end{pmatrix}.
\]
For a matrix in this form, multiplication from the right by $G_i$ has the same effect as multiplication by $fI$ 
so that $A(f^iXG_i)\prod_{j=i+1}^{e-1}G_j = A(f^{i+1}X) \prod_{j=i+1}^{e-1} G_j.$
Finally, since $Af^i$ has $r_i$ non-zero invariant factors, the maximum number of
columns of $Af^iX$ which are zero is equal to $b-r_i$ and we conclude that
$g_i \geq r_i$ and  
$\deg(\det(G)) = \sum_{i=0}^{e-1} g_i d \geq \sum_{i=0}^{e-1} r_i d$.
\end{proof}

\begin{theorem}
\label{thm_pc_if}
Let $f \in \P$ be an irreducible polynomial of degree $d$, $e$ be a positive integer,  
$A \in \P_{f^e}^{b \times b}$, 
and let $G \in \P^{b \times b}$ minimally generate $A$.  Let
\[
\snf(A) = \diag(\underbrace{f^0,\ldots,f^0}_{m_0},\ldots,\underbrace{f^{e-1},\ldots,f^{e-1}}_{m_{e-1}},\underbrace{0,\ldots,0}_{m_e}).
\]
Then
\[
T\snf(G)T = \diag(
\underbrace{f^e,\ldots,f^e}_{m_0},
\ldots,
\underbrace{f,\ldots,f}_{m_{e-1}},
\underbrace{1,\ldots,1}_{m_e}
),
\]
where $T$ is the ones-on-antidiagonal matrix, i.e., we've reversed order of the invariants for convenience.  Moreover, $m_i = r_{i-1} - r_i$, where $r_i$ is the number of non-zero
invariant factors of $Af^i$.
\end{theorem}
\begin{proof}
Observe that $r_i = \sum_{j=0}^{e-i-1} m_j$ for $0 \leq i < e$ 
and consequently $m_i = r_{i-1} - r_i$.  
Let $A = P\snf(A)Q$, where $P,Q \in \P^{b \times b}$ are unimodular, and let
\[
H = Q^{-1} \diag(\underbrace{f^e,\ldots,f^e}_{m_0},\ldots,\underbrace{f,\ldots,f}_{m_{e-1}},\underbrace{1,\ldots,1}_{m_e}).  
\]
By definition, $H$ generates $A$, and $\deg(\det(H)) = \sum_{i=0}^{e-1} r_i d$.  By Lemma \ref{lem:min_gen_bound}, $H$ is minimal because no generators with lower determinantal degree exist.
Since minimal generators are unimodularly equivalent, $G$ has the same Smith form.
\end{proof}

Let $S \in \S_{fg}^{b \times b}$ where $\gcd(f,g) = 1$, $G_1 \in \P_f^{b \times b}$ minimally generate $S$ modulo $f$, and $G_2 \in \P_g^{b \times b}$ minimally generate $S$ modulo $g$.  By the Chinese remainder theorem and Newman's Theorem II.14 \citep{Newman72}, the Smith normal form of the minimal generator, $G \in \P^{b \times b}$, of $S$ is $\snf(G) = \snf(G_1)\snf(G_2)$.  This observation leads to the following theorem, which reduces the probability calculation to primary components.

\begin{theorem}
\label{thm:probprod}
Let $A$ and $B$ be matrices with relatively prime minimal polynomials.  Then $$\Pbr{b}{r}\left(A \oplus B\right) = \Pbr{b}{r}(A) \Pbr{b}{r}(B).$$
\end{theorem}
\begin{proof}
Let $f$ and $g$ be the minimal polynomials of $A$ and $B$ respectively, and let
$S_f = \{U_1^TA^iV_1\}_{i=0}^\infty$ and $S_g = \{U_2^TB^iV_2\}_{i=0}^\infty$ with 
minimal generators $G_1$ and $G_2$ respectively.  Then $fg$ is the minimal polynomial
of $A \oplus B$. Let 
\[
S_{fg} = \left\{(U_1^T U_2^T) (A\oplus B)^i\begin{pmatrix}V_1\\ V_2\end{pmatrix} \right\}_{i=0}^\infty
\]
with minimal generator $G$.
Since $\snf(G) = \snf(G_1) \snf(G_2)$, $G$ is $r$-faithful if and only if 
$G_1$ and $G_2$ are $r$-faithful.
\end{proof}

In view of theorem \ref{thm:probprod} we may focus on each primary component.  Henceforward the matrix will be of the form $\dirsum_{i=1}^m J_i$, where for given irreducible polynomial $f \in \P$ and nonincreasing exponent sequence $e_1, e_2, \ldots, e_m$ we let $J_i$ denote the Jordan block $J_{f^{e_i}}$.

\begin{theorem}
\label{thm:prob_lif}
Using the primary component Jordan form notation just introduced,
let $t$ be the greatest index such that $e_t = e_1$ (thus first index such that $e_t > e_{t+1}$).
For all $r \leq t$,\\
$$\Pbr{b}{r}(\dirsum_{i=1}^m J_i) = \Pbr{b}{r}(\dirsum_{i=1}^{t} J_1)
	 = \Pbr{b}{r}(\dirsum_{i=1}^{t} C_f).$$
[Lower order invariants don't matter, and the effect of a Jordan block is the same as that of a companion matrix.] 
\end{theorem}
\begin{proof}
Let $A = \dirsum_{i=1}^m J_i$, let $G$ minimally generate $S = \{ U A^i V \}_{i=0}^\infty$, 
and let $e = e_1$.  
By Theorem \ref{thm_pc_if}, $G$ is $r$-faithful if 
the number of invariant factors of $\phi_{f^e}\left(Sf^{e-1}\right)$  is at least
$r$.
Because $J_i f^{e-1} = 0$ for all $i \geq t$, $UAf^{e-1}V = \sum_{i=1}^m U_i J_i f^{e-1} V_i = \sum_{i=1}^t U_i J_i f^{e-1} V_i$, where $U_i,V_i$ are blocks of $U,V$ conforming to the blocks of $A$.  
Furthermore, $\{U_i J_i^i f^{e-1} V_i\}_{i=0}^\infty = \{U_{i,e} C_f^i M V_{1,i}\}_{i=0}^\infty$, 
where $M$ is nonsingular, and $U_{i,e}$ and $V_{1,i}$ are the rightmost and topmost blocks of $U_i$ and $V_i$ respectively \citep{HJS:2016}.  
Because $V_{1,i}$ is uniformly random and $M$ is nonsingular, $MV_{1,i}$ is uniformly random, and
therefore, $\Pbr{b}{r}(A) = \Pbr{b}{r}\left(\dirsum_{i=1}^t C_f \right)$.
\end{proof}

Thus we may focus attention on the probability in the case of companion matrices.
To complete the picture 
we will reduce 
to the probability that a sum of outer products has a given rank.  
First (Lemma \ref{lemma_outer})
we observe the relationship between sequences of projections of companion matrices and outer products.  
Then (Theorem \ref{thm:prob_companion})
we relate the sum of outer products to the probability.

\begin{lemma}
\label{lemma_outer}
Let $S = \{U^T C_f^i V\}_{i=0}^\infty$, where $f \in \P$ is an irreducible polynomial of degree $d$, and $U,V \in \F^{d \times b}$ are chosen uniformly at random.  Then $\phi_f(S)$ is the outer product of two uniformly random vectors in $\P_f^b$.
\end{lemma}
\begin{proof}
The $ij$ entry of $S$ satisfies
\[
\omega_f (S_{ij}) = (U_i^TV_j, U_i^TC_fV_j, \ldots, U_i^TC_f^{d-1}V_j)^T =  \rho(V_j)^T U_i, 
\]
where $U_i$ denotes the $i$-th column of $U$ and similarly for $V_j$.
Consequently the $ij$ entry of $\phi_f(S)$ 
\[
\phi_f(S)_{ij} = P \rho(V_j)^T U_i = \rho(V_j)(PU_i) = V_j (PU_i).
\]
Since $P$ is nonsingular and $U_i$ is uniformly random, and $PU_i$ is also uniformly random.  Therefore, $\phi_f(S)$ is the outer product of two uniformly random vectors in $\P_f^b$.
\end{proof}

\begin{definition}
Let $Q_{q,b,r}(t)$ denote the probability that $r = \rank(A)$ when $A$ is a sum of $t$ outer products, $A = \sum_{i=1}^t u_i v_i^T$, and the vectors $u_i,v_i \in \F^b$ are chosen uniformly at random. 
\end{definition}

\begin{theorem}
\label{thm:prob_companion}
For irreducible $f \in \P$ of degree $d$,  
\[\Pbr{b}{r}(\dirsum_{i=1}^t C_f) = \sum_{i=r}^t Q_{q^d,b,i}(t).\]
\end{theorem}
\begin{proof}
Let $S = \{U^T \dirsum_{i=1}^t C_f^iV\}_i$, where $U,V \in \F^{d \times b}$ are chosen at uniformly random.  Let $G \in \P^{b \times b}$ minimally generate $S$.  By Theorem \ref{thm_pc_if}, $G$ has $r$ nontrivial invariant factors, where $r = \rank(\phi_f(S))$.  
By Lemma~\ref{lemma_outer}, $\phi_f(S)$ is the sum of outer products of uniformly random 
vectors in $\P_f^b$, and the probability that $i = \rank(\phi_f(S))$ is $Q_{q^d,b,i}(t)$.
\end{proof}

The probability that the sum of $t$ outer products has rank $r$ can be computed with the following recurrence \citep{HJS:2016}.

\begin{theorem}
\label{thm:probrank}
\[
Q_{q,b,r}(t) = \begin{cases}
  0 & \text{if $r < 0$ or $r > \min(t,b)$,} \\
  1 & \text{if $r = 0$ and $t = 0$,} \\
  \psi_{t,r} & otherwise,
\end{cases}
\]
where $\psi_{t,r} = Q_{q,b,r-1}(t-1) U_{r-1} + Q_{q,b,r}(t-1) (1-U_r-D_r) + Q_{q,b,r+1}(t-1) D_{r+1},$
with $U_r = \left(1-1/{q^{b-r}}\right)^2$ and $D_r = q^{r-1} \left(q^r - 1\right) / q^{2b}$.
\end{theorem}

Now a formula for $\Pbr{b}{r}(A)$ follows from Theorems \ref{thm_pc_if}, \ref{thm:prob_lif} 
and \ref{thm:prob_companion}. A 
Schur complement argument is involved in separating the leading repeated block from the lower exponent blocks.

\begin{lemma}
\label{lem_prob_rec}
Let $f \in \P$ be an irreducible polynomial of degree $d$, and let $e > 1$.  Let $G = I_r \dirsum
 0_{b-r}$ and $H = \left(\begin{matrix} A & B \\ C & D \end{matrix}\right) \in \P_{f^e}^{b \times b}$ be uniformly random, where $f | H$ and $A,B,C,D$ are blocks conforming to the dimensions of the blocks of $G$.  
Then, $G+H \sim (I_r+A) \dirsum Y$, where $Y$ is a  projection of $H$.
\end{lemma}
\begin{proof}
Because $f|A$ and $\P_{f^e}$ is a local ring, the matrix $I_r + A$ is nonsingular.
Let $$P = \left(\begin{matrix}
I & 0 \\ -C(I_r+A)^{-1} & I
\end{matrix}\right)\mbox{ and }Q = \left(\begin{matrix}
I & -(I_r+A)^{-1}B \\ 0 & I
\end{matrix}\right).$$
Then, $$P(G+H)Q = (I_r+A) \dirsum (-C(I_r+A)^{-1}B+D) = (I_r+A) \dirsum Y,$$
where $Y = -C(I_r+A)^{-1}B+D$.
$P$ and $Q$ are trivially unimodular, and $-(I_r+A)^{-1}B$ is the top-right block of $Q$.  Therefore
, $Y$ is a $b-r \times b-r$ projection of $H$.
\end{proof}

\begin{theorem}
\label{thm:probrec}
Let $f \in \P$ be an irreducible polynomial of degree $d$, and   
$A = \dirsum_{i=1}^m J_i$ (notation of theorem \ref{thm:prob_lif}) and let $t$ be the greatest index such that $J_t = J_1$.
Then
\[
\Pbr{b}{r}(A) = \begin{cases}
	1 & \text{if $A = 0$,} \\
        \Pbr{b}{t}(A) \Pbr{b-t}{r-t}(\dirsum_{i=t+1}^m J_i) & \text{if $r > t$,} \\
	\sum_{i=r}^{t} Q_{q^d,b,i}(t) & \text{if $r \leq t$.}
\end{cases}
\]
\end{theorem}
\begin{proof}
The probability that the largest invariant factor is preserved at least $r$ times is given in Theorem \ref{thm:prob_companion}, and by Theorem \ref{thm:prob_lif} is independent of smaller invariant factors being preserved.  Therefore, if $r \leq t$, the probability that the largest $r$ invariant factors are preserved is $\Pbr{b}{r}(\dirsum_{j=1}^{t} J_i) = \sum_{i=r}^{t} Q_{q^d,b,t}(i)$.

If $r > t$, the $t$ largest invariant factors must be successfully preserved along with the
next $r-t$ invariant factors.  
Let $G$ minimally generate $S = \{ U A^i V \}_{i=0}^\infty$.
By Theorem~\ref{thm_pc_if}, if the leading $t$ invariant factors of $G$ are correct if
there are $t$ ones in the Smith normal form of $\phi_{f^e}(S)$.  Therefore, by a unimodular
transformation, $\phi_{f^e}(\dirsum_{i=1}^t J_i)$ is equivalent to a block matrix $I_t \dirsum 0_{b-r}$
in the hypothesis of Lemma~\ref{lem_prob_rec} and consequently the remaining blocks  
$\phi_{f^e}(\dirsum_{i=t+1}^m J_i)$ are projected onto a $(b-t) \times (b-t)$ block.  Since the projection was
accomplished by a unimodular transformation,
the probability that the remaining $r-t$ invariant factors are successfully preserved 
is the probability that a uniformly random $(b-t) \times (b-t)$ projection preserves them.
\end{proof}

\section{Examples}

In this section we present several examples of how to compute $P_{q,b,r}(A)$ for given matrix structures.  Let $q=2$, and let $f,g,h \in \P$ be distinct irreducible with $\deg(f) = \deg(g) = 1$ and $\deg(h) = 2$.  Note that these are the three lowest degree irreducible polynomials in $\P$.  For $M \in \F^{n \times n}$ let $F(M)$ denote the list of invariant factors of $xI-M$.  For example, let $F(A) = \{f^2gh, fg\}$.

To compute $P_{q,b,r}(A)$, first $A$ is split into its distinct factors (Theorem \ref{thm:probprod}) \[
P_{q,b,r}(A) = P_{q,b,r}(C_{f^2} \dirsum C_f) P_{q,b,r}(C_g \dirsum C_g) P_{q,b,r}(C_h).
\]

When $r=1$, applying Theorem \ref{thm:probrec} yields
\begin{eqnarray*}
P_{q,b,r}(A) &=& P_{q,b,1}(C_{f^2}) P_{q,b,1}(C_g \dirsum C_g) P_{q,b,1}(C_h) \\
&=& Q_{q,b,1}(1) (Q_{q,b,1}(1) + Q_{q,b,1}(2)) Q_{q^2,b,1}(1)
\end{eqnarray*}

Otherwise, when $r\geq 2$,
\begin{eqnarray*}
P_{q,b,r}(A) &=& P_{q,b,1}(C_{f^2}) P_{q,b-1,1}(C_f) P_{q,b,2}(C_g \dirsum C_g) P_{q,b,1}(C_h) \\
&=& Q_{q,b,1}(1) Q_{q,b-1,1}(1) Q_{q,b,2}(2) Q_{q^2,b,1}(1)\\
&=& Q_{q,b,2}(2) Q_{q,b,2}(2) Q_{q^2,b,1}(1).
\end{eqnarray*}

\begin{table}
\begin{center}
\caption{$P_{q,b,r}(M_i)$ ({\bf worst case} probability $P_{q,b,r}(8)$), $q = 2$}
\begin{tabular}{|c||c|c|c|c|c|c|c|c|} \hline
& \multicolumn{2}{c|}{r=2} & \multicolumn{2}{c|}{r=3} & \multicolumn{2}{c|}{r=4} & \multicolumn{2}{c|}{r=5}           \\ \hline
      & b=2         & b=3         & b=3         & b=4         & b=4         & b=5         & b=5         & b=6         \\ \hline
$M_1$ & {\bf 0.010} & {\bf 0.158} & 0.158       & 0.435       & 0.435       & 0.674       & 0.674       & 0.825       \\
$M_2$ & 0.053       & 0.439       & {\bf 0.011} & {\bf 0.142} & 0.142       & 0.398       & 0.398       & 0.639       \\
$M_3$ & 0.095       & 0.625       & 0.041       & 0.393       & {\bf 0.009} & {\bf 0.126} & 0.126       & 0.374       \\
$M_4$ & 0.084       & 0.570       & 0.028       & 0.286       & 0.069       & 0.382       & {\bf 0.056} & {\bf 0.275} \\
$M_5$ & 0.070       & 0.367       & 0.367       & 0.647       & 0.647       & 0.817       & 0.817       & 0.907       \\ \hline
\end{tabular}
\label{tab:probs}
\end{center}
\end{table}

Let $M_1, M_2, M_3, M_4, M_5 \in \F^{8 \times 8}$, $F(M_1) = \{fgh,fgh\}$, $F(M_2) = \{fgh, fg, fg\}$, $F(M_3) = \{fg,fg,fg,fg\}$, $F(M_4) = \{fg,fg,fg,f,f\}$, and $F(M_5) = \{f^2h, f^2h\}.$  To illustrate the effect of invariant structure on $P_{q,b,r}(M)$, Table \ref{tab:probs} shows $P_{q,b,r}(M_i)$ computed for $b = \{2\ldots 6\}$ and $r=\{2\ldots 5\}$.  Note that $M_1, M_2, M_3,$ and $M_4$ are the worst case matrices for $r=2,3,4,$ and $5$, respectively, using the worst case construction given in the following section.

We also performed an experimental check on the probabilities $P_{3,5,r}(M)$.
In Novocin et al. (2015) the Ding-Yuan family of matrices were among those studied, with the goal of developing a formula for their ranks over the field $\F_3$. 
One matrix in this family is $M \in \F_3^{27\times 27}$ 
having invariants $\{x,\ldots,x,xf, xfg\}$ wherein $x$ appears 19 times and $f=x^2+x, g=x^4 + x^3 + x + 2$.  
We computed with $b = 5$ 
and obtained the results in table \ref{tab:thousand}, giving the probability 
$P_{3,5,r}(M)-P_{3,5,r+1}(M)$ that exactly $r$ invariants are correct and for each $r$ the percentage of the ten thousand trials in which exactly $r$ correct invariants resulted.  The data is quite consistent with theory.
It turns out that 3 or more correct invariants is sufficient in this case to infer the rank.  Only two trials failed to provide the first 3 invariants. 
The example $M$ is further discussed in section \ref{sec:post-hoc-ex}.
\begin{table}
\begin{center}
\caption{Ten thousand trials, $b = 5$, $M$ has 19 invariants: $(x,\ldots,x,xf,xfg)$}
\begin{tabular}{|c||c|c|c|c|} \hline
$b$ invariants & $x,x,x,xf,xfg$ & $1,x,x,xf,xfg$ & $1,1,x,xf,xfg$ & $1,1,1,*,*$
 \\
$r$ = number correct & 5 & 4 & 3 & 2,1,0 \\ \hline
probability & 56.11 & 41.91 & 1.94 & 0.04 \\ \hline
\% of trials& 55.80 & 42.08 & 2.10 & 0.02 \\ \hline
\end{tabular}
\label{tab:thousand}
\end{center}
\end{table}
\section{Worst Case}
\label{sec:wc}

Recall from the introduction that we 
define $P_{q,b,r}(n) = \min_{A \in \F^{n \times n}} \Pbr{b}{r}(A)$.
The formula we will derive for $P_{q,b,r}(n)$ can be used to determine the necessary blocksize needed to preserve the leading $r$ invariant factors with a specified probability of success.  It will show that with a blocksize modestly larger than $r$ the probability of preserving $r$ invariant factors is quite high, even for small fields.  The construction and formula generalize Theorem 20 from \citep{HJS:2016} which obtained a similar bound for preserving the minimal polynomial.  
To develop the formula, we begin with the following properties derived from
Theorems~\ref{thm:probprod} and \ref{thm:probrec} to compute the probability for the 
leading Jordan block and the Schur complement to induct on the remaining blocks.

\begin{lemma} 
\label{lem:tech}
\mbox{}
\begin{enumerate}
\item $\Pbr{b}{r}( \underbrace{C_f \oplus \ldots \oplus C_f}_r ) \leq \Pbr{b}{r}( \underbrace{C_f \oplus \ldots \oplus C_f}_t )$ for $r < t$.
\item Let $f$ and $g$ be irreducible polynomials of degree
$d$ and $e$ respectively with $d < e$, then \\
$\Pbr{b}{r}( \underbrace{C_f \oplus \ldots \oplus C_f}_r )
< \Pbr{b}{r}( \underbrace{C_g \oplus \ldots \oplus C_g}_r )$.
\item Let $f$ and $g$ be irreducible polynomials both of degree
$d$ and let r be given.  Then \\
$\Pbr{b}{s}( \underbrace{C_f \oplus \ldots \oplus C_f}_s \oplus \underbrace{C_g \oplus \ldots \oplus C_g}_t )$, for $s + t = r$ is minimized when $s = r$ and $t=0$.
\end{enumerate}
\end{lemma}
\begin{proof}
Parts 2 and 3 follow from Theorem \ref{thm:probrec} and are straightforward.  Part 1, while intuitively clear, is more complicated.  For part 1, let $m = q^d$, where $d = \deg(f)$.  Let $P_t(r) = P_{q,b,r}(\dirsum_{i=1}^t C_f)$, and let $U_r, D_r$ be defined as in Theorem \ref{thm:probrank}.  Note that $U_r$ and $D_r$ denote the probability that given a matrix $A \in \P_f^{b \times b}$ of rank $r$ and random vectors $u,v \in \P_f^b$ that $\rank(A + uv^T) = r+1$ and $r-1$, respectively.  By Theorem 15 of \cite{HJS:2016}, $P_1(1) \leq P_t(1)$ for all $t \geq 1$.  Applying induction, for all $t \geq r \geq 2$,
\begin{eqnarray*}
P_{t+1}(r) &=& P_t(r) + Q_{m,b,t}(r-1)U_{r-1} - Q_{m,b,t}(r)D_r \\
&=& P_t(r) + (P_t(r-1) - P_t(r))U_{r-1} - Q_{m,b,t}(r)D_r \\
&\geq & P_t(r) + (P_t(r-1) - P_t(r))U_{r-1} - P_t(r)D_r \\
&=& P_t(r-1)U_{r-1} + P_t(r)(1 - U_{r-1} - D_r) \\
&\geq & P_t(r-1)U_{r-1} \\
&\geq & P_{r-1}(r-1)U_{r-1} \\
&=& P_r(r),
\end{eqnarray*}
because $Q_{m,b,t}(r) \leq P_t(r)$ and $
1 - U_{r-1} - D_r = \frac{2q^{b+r+1} - q^{2r+1} - q^{2r} + q^{r+1}}{q^{2b+2}} > 0$.
\end{proof}

Property 1 implies that, for a given irreducible and in the absence of other variation, the probability is increased if more
than $r$ copies of $C_f$ are included and property 3 implies
that the probability is increased if fewer than $r$ copies of
$C_f$ are included.  Property 2 implies that the probability is
larger if higher degree irreducibles are used.  Therefore the
probability that the $b \times b$ projection of 
$A \in \F^{n \times n}$ preserves the leading $r$ invariant factors is minimized
when $A$ is a direct sum of companion matrices for distinct irreducible polynomials, each repeated r times (to the extent possible) and of the smallest degrees, subject only to the requirement that the sum of the degrees is $n$. 
The construction begins by including $r$ copies of companion matrices for
each irreducible over $\F$ of degree one, then $r$ copies of companion matrices for each irreducible of degree two over
$\F$, and so on until the sum of the degrees is equal to $n$.  It may be the case that fewer than $r$
copies of the last irreducible fit and in this case the probability is minimized by including as many copies as do fit. Moreover, when filling in the last batch of companion
matrices of degree $m$ the dimension $n$ may not be reached
exactly and in this case several companion matrices can be
replaced by Jordan blocks to fit the dimension without changing
the probability.

Using this construction, a formula for $P_{q,b,r}(n)$ can
be derived.  Let $L_q(m)$ to be the number of monic irreducible polynomials of degree $m$ in $\P$ and define 
$L_{q,r}(n,m) = \min\left(L_q(m),\left\lfloor 
\frac{s}{rm} \right\rfloor\right),$ for 
$s = n - \sum_{d=1}^{m-1} r d L_q(d).$

\begin{figure}[t]
\begin{center}
\includegraphics[scale=0.46, trim=20 30 20 20]{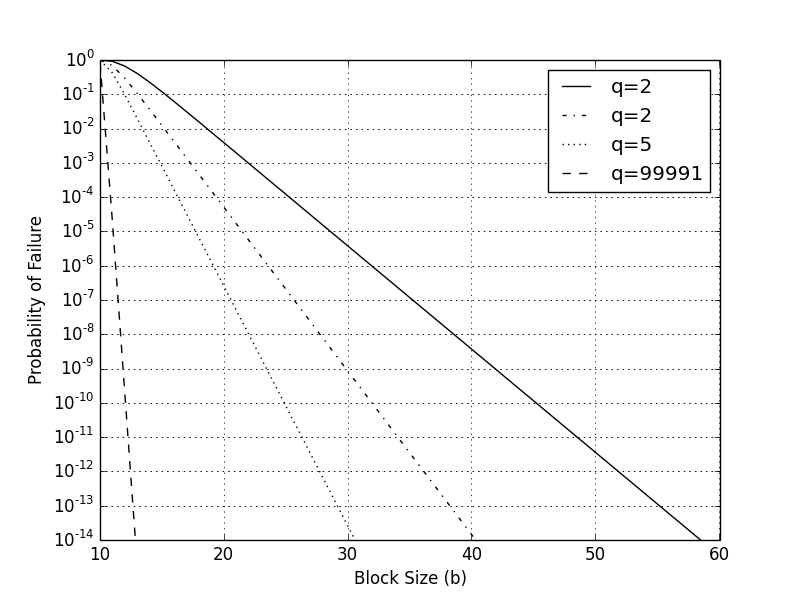}
\caption{Worst case probability of failure ($1-P_{q,b,10}(10^8)$) vs cardinality and block size}
\label{plot:worst-case}
\end{center}
\end{figure}

\begin{theorem}
\label{theorem:worst}
Define $m$ such that $\sum_{d=1}^{m-1} rd L_q(d) \le n$ and $n < \sum_{d=1}^m rd L_q(d)$.  Let  
$s = n - \sum_{d=1}^{m} r d L_{q}(n,d)$ and
$t = \left\lfloor \frac{s}{m} \right\rfloor $.  
Then the worst case probability that $r$ invariants are preserved in projection to blocksize $b$ on an $n\times n$ matrix over the field of cardinality $q$ satisfies 
\begin{eqnarray*}
&&P_{q,b,r}(n) = \left( \prod_{d=1}^m Q_{q^d,b,r}(r)^{L_{q,r}(n,d)} \right) Q_{q^m,b,t}(t).
\end{eqnarray*}
\end{theorem}
\begin{proof}
Let $A \in \F^{n \times n}$ be the worst case matrix described by the construction above, following from Lemma \ref{lem:tech}.  Applying Theorem \ref{thm:probrec}, $P_{q,b,r}(n) = P_{q,b,r}(A) =  \left( \prod_{d=1}^m Q_{q^d,b,r}(r)^{L_{q,r}(n,d)} \right) Q_{q^m,b,t}(t)$.
\end{proof}

The plot in Figure \ref{plot:worst-case} shows the worst case probability of failure, 
($1-P_{q,b,10}(10^8)$), versus block size, $b$, for the field cardinalities 2,3,5, and 9991.
This shows that with a blocksize of a little over 30 the probability that the first
ten invariants are not correct is approximately one in a million, 
with substantially better probabilities for larger fields.

\section{Bound}

The worst case given in Theorem \ref{theorem:worst} is tight but not easily computed.  A simplified bound is offered in Theorem \ref{thm:bound}.

\begin{theorem}
\label{thm:bound}
For any $n > 0$, $q \geq 2$, and $b > r > 0$,
\[
P_{q,b,r}(n) > 1 - \frac{2r}{q^{b-r}-1}.
\]
\end{theorem}
\begin{proof}
\begin{eqnarray*}
P_{q,b,r}(n) &>& \prod_{d=1}^\infty \prod_{i=0}^{r-1} \left( 1 - \frac{1}{q^{d(b-i)}} \right)^{2L_{q,r}(d)} 
\geq \prod_{d=1}^\infty \left( 1 - \frac{1}{q^{d(b-r+1)}} \right)^{2rq^d} \\
&>& \prod_{d=1}^\infty \left( 1 - \frac{2r}{q^{d(b-r)}} \right) 
> 1 - \sum_{d=1}^\infty \frac{2r}{q^{d(b-r)}} \\
&=& 1 - \frac{2r}{q^{b-r} - 1}
\end{eqnarray*}
\end{proof}

Theorem \ref{thm:bound} implies that taking $b = \left\lceil\log_q\left(\frac{2r}{1-p}+1\right)+r \right\rceil$ will yield $P_{q,b,r}(n) > p$.  As discussed in Section \ref{sec:wc}, there exists a matrix $A \in \F^{n \times n}$ such that $xI-A$ has at most $r$ nontrivial invariant factors and $P_{q,b,r}(n) = \Pbr{b}{r}(A)$.  Thus the bounds given in \citep{Brent:2003} for matrices with $r$ or fewer nontrivial invariants 
apply as well to $\Pbr{b}{r}(A)$.

\begin{theorem}
\label{theorem:brent}
For any $n > 0$, $q \geq 2$, and $b > r > 0$,
the probability that $r$ invariants are preserved in projection to blocksize $b$ on an $n\times n$ matrix over the field of cardinality $q$ satisfies 
\[
P_{q,b,r}(n) \geq \begin{cases}
\frac{1}{64} & \text{if $b = r+1$ and $q=2$,} \\
\left( 1 - \frac{3}{2^{b-r}} \right)^2 \geq \frac{1}{16} & \text{if $b \geq r+2$ and $q=2$,} \\
\left( 1 - \frac{2}{q^{b-r}} \right)^2 \geq \frac{1}{9} & \text{if $b \geq r+1$ and $q > 2$.}
\end{cases}
\]
\end{theorem}
\begin{proof}
By Theorem \ref{theorem:worst} and its proof, there exists a matrix $A \in \F^{n \times n}$ such that $A$ has at most $r$ invariant factors and $P_{q,b,r}(A) = P_{q,b,r}(n)$.  Therefore, the bound (squared) given in 
\cite{Brent:2003} applies to $P_{q,b,r}(n)$.
\end{proof}

Theorem \ref{theorem:brent} implies that taking $b = \left\lceil \log_2 \left( \frac{3}{1 - \sqrt{p}} \right) + r \right\rceil$ when $q = 2$, and $b = \left\lceil \log_q \left( \frac{2}{1 - \sqrt{p}} \right) + r \right\rceil$ when $q > 2$ will yield $P_{q,b,r}(n) \geq p$.

\section{Post hoc analysis}
\label{sec:post-hoc}
In the block Wiedemann algorithm applied with blocksize $b$ to matrix $A \in \F^{n\times n}$, we compute candidates for the leading $b$ invariant factors.  
The computed invariants list divides the true leading invariant list componentwise.  
In the preceding sections we have been concerned with the a priori probability of getting the true leading invariants.  
Here we consider what we can know post hoc.  
For example, if the degrees of the computed invariants total $n$, then with certainty, we have computed the true invariants.
At the other extreme, 
if $n$ is odd and all irreducibles in the factors we found are of even degree, then with certainty we have missed a factor in the first invariant, the minimal polynomial.  
There is a great variety of cases between these extremes.  For example if the $b$ computed invariants are a nontrivial minimal polynomial followed by $x$'s, as often occurs for low rank matrices, the chance of error is low, since any error must include a missing factor in at least one of the leading two invariants.  The probability of this diminishes rapidly as $b$ grows.
For a more problematic example, 
suppose $b < n/2$ and the computed invariants are $x^2-x, x^2-x, \ldots, x^2-x, x$.  
In other words, $x-1$ is a computed factor $b-1$ times and $x$ is computed $b$ times. 
Since each invariant divides the preceding, the possibilities of location of the first error are limited to the first invariant and the $b$-th.  
We may have missed an extra factor in the first position, but with very low probability (depending on $b$).  
Also we may have missed a factor $x-1$ in the $b$-th invariant, a much more likely scenario.  
In this case, with high confidence there are no irreducibles occurring other than $x-1$ and $x$ and no Jordan blocks for those two factors of exponent greater than 1.  
On the other hand, confidence in knowing the rank of the matrix is lower, as that depends on believing that $x-1$ does not divide the $b$-th invariant.

Many of these observations can be derived using the a priori bound, $P_{q,b,r}(n)$. 
The chance of incorrectly computing the minimal polynomial is small, even when computation on a large matrix over a small field using a modest block size, e.g., $P_{2, 10, 1}(10^9) \approx 0.996$.  
The a priori bound also shows that computing the second and third invariants correctly is also very likely, $P_{2, 10, 2}(10^9) \approx 0.988$ and $P_{2, 10, 3}(10^9) \approx 0.973$.
The confidence given by the a priori bound decreases dramatically as $r$ approaches $b$.
Reusing the previous example, suppose the computed factors are $x^2-x, \ldots, x^2-x, x$.  Using the a priori bound, the probability that the $b$-th invariant factors was correctly computed is at least $P_{2,10,10}(10^9) \approx 0.00004$.  

Confidence in computed results are improved significantly by applying post hoc analysis.
The remainder of this section describes a practical framework for doing post hoc analysis.

To develop these observations into post hoc bounds, we begin with some definitions.
Let $P_{q,b,r}(G, A)$ denote the probability that block Wiedemann would produce the first $r$ invariant factors of $G$ given $A$, where there exists $U,V \in \F^b$ such that $G$ minimally generates $\{U^TA^iV\}_i$.  Let $P_{q,b,r}(G, n)$ be a lower bound on the probability that $G$ was computed $r$-faithfully to $xI-A$ where $A \in \F^{n \times n}$ is unknown.  
Let $g_i$ and $a_i$ denote the $i$-th largest invariant factor of $G$ and $xI-A$, respectively.
Let $\mathcal{A}$ be the set of all $A \in \F^{n \times n}$ such that $g_i | a_i$ and $G$ is not $r$-faithful to $xI-A$.  Then, $
P_{q,b,r}(G, n) = 1 - \max(\{P_{q,b,r}(G,A) ~|~ A \in \mathcal{A}\}).
$  
Note that $\mathcal{A} = \emptyset$ when $\deg(\det(G)) = n$ and therefore $P_{q,b,r}(G, n) = 1$.



In practice, computing $P_{q,b,r}(G, n)$ is prohibitively time consuming because of the cardinality of $\mathcal{A}$.  Therefore, the rest of this section is dedicated to describing a lower bound on $P_{q,b,r}(G, n)$.  We do this by locating the most likely error.  The first lemma observes that each primary component can be considered separately.

\begin{lemma}
\label{lem:post-hoc-if}
Let $f_i \in \P$ be distinct irreducible polynomials.  
Let $A = \dirsum_i A_i$ where $f_i^{e_i} = \minpoly(A_i)$.  
Let $\snf(G) = \prod_i \snf(G_i)$, where $G_i$ minimally generates $\{U_i^TA_i^kV_i\}_k$.  Then, \[
P_{q,b,r}(G, A) = \prod_{i=1}^k P_{q,b,r}(G_i, A_i) \leq \max_i P_{q,b,r}(G_i, A_i).
\]
\end{lemma}
\begin{proof}
The equality follows from Theorem \ref{thm:probprod}.  
The inequality is trivially true, because $0 < P_{q,b,r}(G_i, A_i) < 1$.
\end{proof}

In the next lemma, we give a bound on $P_{q,b,r}(G, A)$ for a single primary component and a candidate matrix $A$ where $G$ is $i$-faithful but not $(i+1)$-faithful.  The bound is then used to compute a bound on the probability that the first error occurred in $G$ at index $i+1$ for unknown input.

\begin{lemma}
\label{lem:post-hoc-prob1}
Let $S = \{U^TA^iV\}_i$, where $A \in \F^{n \times n}$, and $U, V \in \F^{n \times b}$ are uniformly random.  
Let $G \in \P^{b \times b}$ minimally generate $S$.  
Let $\minpoly(A) = f^e$, where $f \in \P$ is an irreducible polynomial of degree $d$.  
Let $G$ be $i$-faithful, but not $(i+1)$-faithful, where $i < r$.  
Let $j$ and $k$ be the first and last indices, respectively, such that $a_j = a_{i+1} = a_k $.  
Then, \[
P_{q,b,r}(G, A) \leq Q_{q^d, b, j-1}(j-1) Q_{q^d, b - j + 1, i - j + 1}(k - j + 1).
\]
\end{lemma}
\begin{proof}
By definition, the $Q_{q^d, b, j-1}(j-1)$ is the probability that $G$ is $(j-1)$-faithful to $xI-A$.  Following from Theorem \ref{thm:probrec}, $Q_{q^d, b - j + 1, i - j + 1}(k - j + 1)$ is the probability that exactly $i - j + 1$ copies of $a_{i+1}$ were preserved by random $(b - j + 1) \times (b - j + 1)$ projection.  
The probability, $p$, that the remaining invariant factors of $G$ would be computed is left out to simplify the calculation and reduce the search space.  
Therefore, \begin{eqnarray*}
P_{q,b,r}(G, A) &=& Q_{q^d, b, j-1}(j-1) Q_{q^d, b - j + 1, i - j + 1}(k - j + 1) p \\
&\leq & Q_{q^d, b, j-1}(j-1) Q_{q^d, b - j + 1, i - j + 1}(k - j + 1),
\end{eqnarray*}
because $0 < p \leq 1$.
\end{proof}

Let $f \in \P$ be an irreducible polynomial of degree $d$, and let $\det(G) = f^e$.  Let $M_{q,n}(G, i)$ denote the probability that the first error in $G$ occurs at the $(i+1)$-st largest invariant factor.  Let $\mathcal{A}_i$ denote the set of all matrices $A \in \F^{n \times n}$ such that $\minpoly(A) = f^k$.
Then, $M_{q,n}(G, i) = \max_{A \in {\mathcal{A}_i}} P_{q,b,r}(G, A)$  where $g_j = a_j$ for all $j \leq i$, $g_{i+1} \neq a_{i+1}$, and $i < r$.
The following lemma gives an upper bound on $M_{q,n}(G, i)$ by applying Lemma \ref{lem:post-hoc-prob1} to all matrices in $\mathcal{A}_i$.

\begin{lemma}
\label{lem:post-hoc-boundi}
Let $f \in \P$ be an irreducible polynomial of degree $d$.  Let $G \in \P^{b \times b}$, where the $i$-th largest invariant factor of $G$ is $g_i = f^{e_i}$.  Let $m = (n - \deg(\det(G)))/d$.  Then, 
\[
	M_{q,n}(G, i) \leq
	\begin{cases}
	0, & g_i = g_{i+1} \\
	\max(X), & g_i / g_{i+1} = f, \\
	\max(X ~ \bigcup ~ Y), & otherwise,
	\end{cases}
\]
where $X = \{Q_{q^d, b, k}(k) Q_{q^d, b - k, i - k}(i - k + t) ~|~ t = 1...m\}$,
$Y = \{Q_{q^d, b, i}(i) Q_{q^d, b - i, 0}(t) ~|~ t = 1...m\}$, and $k$ is the least index such that $g_{k+1} = g_i$.
\end{lemma}
\begin{proof}
Let $A \in \F^{n \times n}$, and let $a_i$ denote the $i$-th largest invariant factor of $xI-A$.  Let $G$ be $i$-faithful but not $(i+1)$-faithful to $xI-A$.

If $g_i = g_{i+1}$ then it is impossible for $G$ to be $i$-faithful and not $(i+1)$-faithful to $xI-A$, because $a_{i+1} | a_i$ and $g_j | a_j$ for all $j$.  Therefore, if $G$ is $i$-faithful to $xI-A$ and $g_i = g_{i+1} = a_i$ then $a_{i+1} = a_i = g_{i+1}$ and $G$ is $(i+1)$-faithful to $xI-A$.

The set $X$ results from the application of Lemma \ref{lem:post-hoc-prob1} to all $A$ such that $a_{i+t} = a_i$ for all $1 \leq t \leq m$.

The set $Y$ results from the application of Lemma \ref{lem:post-hoc-prob1} to all $A$ such that $a_{i+1} \neq a_i$ and $g_{i+1} \neq a_{i+1}$.  That is, $G$ missed the factor $a_{i+1}$ entirely, which is only possible if $e_i - e_{i+1} > 1$.
\end{proof}

Note that Lemma \ref{lem:post-hoc-boundi} applies to $i = 0$, because all generators are $0$-faithful and there is a chance that $G$ is not $1$-faithful.  Also, $m$ as defined in Lemma \ref{lem:post-hoc-boundi} is given for the purposes of bounding the number computations.  For almost all situations, using $t = 1$ for the calculation gives the minimum probability for a given $i$.

Applying Lemma \ref{lem:post-hoc-if} means that a lower bound for $P_{q,b,r}(G, n)$ can be computed by computing a bound for each primary component and taking the max, applying Lemma \ref{lem:post-hoc-boundi} as described in the following theorem.

\begin{theorem}
\label{thm:post-hoc}
Let $f \in \P$ be an irreducible polynomial of degree $d$.  Let $G \in \P^{b \times b}$, and let $\det(G) = f^e$.  Then, 
\[
P_{q,b,r}(G, n) \geq 1 - \max(\{M_{q, n}(G, i) ~|~ 0 \leq i < r\}).
\]
\end{theorem}

As in the earlier example, let $g_1 =\ldots =g_9 = x^2+x$ and $g_{10} = x$.  Let $H$ be the primary component of $G$ associated with $x+1$.
There are two possible places where the first error could occur, at the $h_1$ and $h_{10}$.
Applying Lemma \ref{lem:post-hoc-boundi},
$M_{2,10^9}(H, 0) \approx 0.002$ and 
$M_{2,10^9}(H, 10) \approx 0.578$.
Therefore, by Theorem \ref{thm:post-hoc}, $P_{2,10,10}(G,10^9) \geq 1 - 0.578 = 0.422$.  
This is a dramatic improvement over the a priori bound, $P_{2,10,10}(10^9) \approx 0.00004$.

\section{Post-Hoc Example}
\label{sec:post-hoc-ex}

In \cite{NSSY15} the Ding-Yuan family of matrices were among those studied, with the goal of developing a formula for their ranks over the field $\F_3.$
We discuss two examples from that sequence, showing that the rank is learned with high confidence by using block Wiedemann with no preconditioning.
Here $q = 3, n = 27$ and $A \in \F^{n \times n}$ is as defined in \cite{NSSY15}.  If we set out to compute the first $r = 3$ invariant factors of $A$ and chose a block size of $b = 5$, we may calculate that $P_{q,b,r}(A) \geq P_{q,b,r}(n) \approx 0.715$.  Selecting random $U, V \in \F^{n \times b}$, we computed $G$ such that $G$ minimally generates $S = \{U^T A^i V\}_i$, and leading invariant factors of $\snf(G)$ are $(xfg, xf, x, x, 1)$, where $f = x^2 + 1,$ and $g = x^4 + x^3 + x + 2$.  
$G$ is clearly not 5-faithful since the total degree of the invariants is 12 while the total degree of the invariants of $xI-A$ will be 27.  But it will suffice 
that $G$ be 3-faithful
to infer the the full list of invariants and thus the desired rank.  The a priori probability of 3-faithfulness is 0.715.
We are able to improve our confidence in the faithfulness of $G$ using the post-hoc analysis developed in the Section \ref{sec:post-hoc}.

The candidates for first errors are: missing a higher power of $x$, $f$, or $g$ in the first invariant, missing $g$ in the second, or missing $f$ in the third.  Bounds on those possibilities are,
\begin{eqnarray*}
M_{q, n}(G_x, 0) &\leq& \max(\{Q_{q, b, 0}(0) Q_{q, b, 0}(t) ~|~ t = 1...15\}) = 0.0082 \\
M_{q, n}(G_f, 0) &\leq& \max(\{Q_{q^2, b, 0}(0) Q_{q^2, b, 0}(t) ~|~ t = 1...7\}) = 0.0000338 \\
M_{q, n}(G_g, 0) &\leq& \max(\{Q_{q^4, b, 0}(0) Q_{q^4, b, 0}(t) ~|~ t = 1...3\}) = 0.000000000573\\
M_{q, n}(G_g, 1) &\leq& \max(\{Q_{q^4, b, 0}(0) Q_{q^4, b, 1}(1 + t) ~|~ t = 1...3\}) = 0.000000047 \\
M_{q, n}(G_f, 2) &\leq& \max(\{Q_{q^2, b, 0}(0) Q_{q^2, b, 2}(2 + t) ~|~ t = 1...7\}) = 0.0031
\end{eqnarray*}

Therefore, the probability that the 3 largest invariant factors of $G$ were computed unfaithfully from any matrix $A \in \F^{n \times n}$ is at most $0.0082$, and our confidence in $G$ being 3-faithful is at least $P_{q,b,r}(G,n) \geq 0.9918$, compared to $P_{q,b,r}(n) \approx 0.715$ a priori.

Taking another example from the Ding-Yuan family.
we have $n = 3^9$ and $A \in \F^{n \times n}$.  We chose a block size of $10$ hoping to get as many as $9$ invariant factors of $xI-A$ correct,
$P_{3,10,9}(n) \approx 0.314$.  Let $G$ be the result of the block Wiedemann computation with $\snf(G) = ( x a^2 b^2 c d, x a^2 b^2 c, x a^2 b^2, x a^2 b^2, x, \ldots, x, 1, 1)$ where $a,b,c,d \in \P$ are distinct irreducibles and $\deg(a) = \deg(b) = 1$, $\deg(c) = 537$, and $\deg(d) = 139$.  
Because the 5th invariant factor of $G$ is irreducible, if we are confident that $G$ is 5-faithful to $A$, then we are confident that we have computed the entire eigenstructure of $A$.  
The worst case is that $G$ misses $a$ or $b$ in the fifth invariant. Thus the probability of 5-faithfulness is at least $P_{3,10,5}(G,n) \geq 1-Q_{3,10,5}(4) \approx 0.996$.  So in fact we can be this confident that the two 1's should be $x$'s. Note also that in this case, since 5 invariants suffice to determine the entire structure, we do nearly as well using the a priori bound
for 5-faithfulness: $P_{3,10,5}(n) \approx 0.988$.

\section{Conclusion}
We have extended the tight bounds of \cite{Brent:2003} for probability of faithfulness of the invariants of a projected sequence so that the bounds apply regardless of the total number of invariants (theorem \ref{theorem:brent}).
The strategy was to exactly compute the probability as a function of the invariant factor structure of the given matrix (theorem \ref{theorem:worst}).  
We can compute the necessary block size required to obtain any desired confidence of obtaining the leading $r$ invariant factors. We show, using block size slightly larger than $r$, that the leading $r$ invariant factors are correct with high probability over any field.

In addition to the a priori bounds, 
we have provided tools for calculating sharper bounds when something is known about the invariants.  
A post hoc analysis (Section \ref{sec:post-hoc}) can often assert a much higher probability of correctness than is available a priori.  
Conditions sufficient for a strong post hoc analysis --- such as few nontrivial invariants --- often hold in practice. When that doesn't apply, preconditioning can be used.  
However the preconditioners that have been proposed and analyzed
focus on the relation of the minimal polynomial to the characteristic polynomial.  
Invariably they require a large field for high worst case probability of success 
(all invariants beyond the minimal polynomial equal to $x$ or 1). 
Observe that, in view of the analysis here, block Wiedemann can be used over small fields and with no preconditioner or a preconditioner used heuristically.  If the condition of few nontrivial invariants is achieved, the method is successful even with very modest block size used.

Consider, for example, a standard rank algorithm: use diagonal preconditioning and apply Wiedemann (block size of one) \citep{CEKTSV02}.  The result is a minimal polynomial of the form $xf(x)$ where $f(0) \neq 0$.  The expectation is that this minpoly is a shift of the characteristic polynomial so that the rank is $\deg(f)$.  
The algorithm works very well over large fields but fails utterly over $\F_2$.  Blocking increases the probability of a correct minpoly,
removes the requirement to believe that charpoly is $x^e$ times minpoly, 
and does not require any proof of the efficacy of the preconditioner (when successful as determined by post hoc analysis).  This is very helpful in the typical situation where the given matrix is far from a worst case.

Similar points apply to other linear algebra problems.  For example, determinant, rank, solving nonsingular and singular consistent systems, nullspace random sampling and nullspace basis are all problems that can be usefully attacked using block Wiedemann, very often with small block size and little or no preconditioning.  Strong probability of success can be had for any field size and without relying on preconditioner worst case probability analysis.

\bibliographystyle{plainnat}
\bibliography{linalg,issac}

\end{document}